\documentclass[conference]{IEEEtran}
\IEEEoverridecommandlockouts
\usepackage{cite}
\usepackage{amsmath,amssymb,amsfonts}
\usepackage{algorithmic}
\usepackage{graphicx}
\usepackage{textcomp}
\usepackage{xcolor}

\usepackage{makeidx} 
\usepackage{amsthm}
\usepackage{algorithm}
\usepackage{relsize}
\usepackage{subfigure}
\usepackage{url}
\usepackage{cite}
\usepackage[utf8]{inputenc}
\usepackage[english]{babel}
\usepackage{longtable}
\usepackage{tabularx}
\usepackage{ltablex}
\usepackage{array}\usepackage{wrapfig}
\usepackage[utf8]{inputenc}
\usepackage{tikz}
\usetikzlibrary{positioning}
\usepackage{mathrsfs}
\usepackage{ragged2e}
\usepackage{verbatim}

\newtheorem{theorem}{Theorem}
\newtheorem{lemma}[theorem]{Lemma}

\newtheorem{fact}[theorem]{Fact}

\newtheorem{construction}[theorem]{Construction}

\theoremstyle{definition}

\newtheorem{example}[theorem]{Example}

\theoremstyle{remark}
\newtheorem{remark}[theorem]{Remark}

\def\BibTeX{{\rm B\kern-.05em{\sc i\kern-.025em b}\kern-.08em
    T\kern-.1667em\lower.7ex\hbox{E}\kern-.125emX}}

\makeatletter
\newcommand*{\rom}[1]{\expandafter\@slowromancap\romannumeral #1@}
\makeatother

\begin{document}

\title{Bounds on Fractional Repetition Codes using Hypergraphs
  \thanks{The part of the work is submitted to ITW 2018.}
 }

\author{\IEEEauthorblockN{Krishna Gopal Benerjee}
\IEEEauthorblockA{\textit{Laboratory of Natural Information Processing,} \\
\textit{Dhirubhai Ambani Institute of Information} \\ \textit{ and Communication Technology,}\\
Gandhinagar, Gujarat, India \\
krishna\_gopal@daiict.ac.in}
\and
\IEEEauthorblockN{Manish K Gupta}
\IEEEauthorblockA{\textit{Laboratory of Natural Information Processing,} \\
\textit{Dhirubhai Ambani Institute of Information} \\ \textit{and Communication Technology,}\\
Gandhinagar, Gujarat, India \\
mankg@computer.org}
}

\maketitle

\begin{abstract}
In the \textit{Distributed Storage Systems} (DSSs), an encoded fraction of information is stored in the distributed fashion on different chunk servers. 
Recently a new paradigm of \textit{Fractional Repetition} (FR) codes have been introduced, in which, encoded data information is stored on distributed servers, where encoding is done using a \textit{Maximum Distance Separable} (MDS) code and a smart replication of packets. 
In this work, we have shown that an FR code is equivalent to a hypergraph. 
Using the correspondence, the properties and the bounds of a hypergraph are directly mapped to the associated FR code. 
In general, the necessary and sufficient conditions for the existence of an FR code is obtained by using the correspondence.
Some of the bounds are new and FR codes meeting these bounds are unknown.
It is also shown that any FR code associated with a linear hypergraph is universally good. 

\end{abstract}


\section{Introduction}\label{intro}
The distributed storage is a well-studied area, which deals with storing the data on distributed nodes in such a manner so that it allows data accessibility at any time and anywhere. 
Many companies such as Microsoft and Google etc. provide such storage services by using distributed data centers.
In the storage systems like Google file system, 
multiple copies of data fragments are stored.
Thus, the system becomes reliable and the file can be retrieved from the system. 
At the same level of redundancy, coding techniques can improve the reliability of the storage system. 

In \cite{5550492}, Dimakis \textit{et al}. proposed a coding scheme called \textit{regenerating codes} which reduces the repair bandwidth, where repair bandwidth is the total number of communicated packets to repair a failed node. 
In the distributed storage system (DSS), encoded data packets are distributed among $n$ nodes such that any data collector can get the complete file by collecting packets from any $k$ $(\leq n)$ nodes. 
Also, a failed node can be repaired by replacing it with a new node having the at-least same amount of the data information. 
The newcomer node is constructed by downloading packets from any $d$ $(\leq n)$ active nodes in the DSS.
The Minimum Bandwidth Regenerating (MBR) codes are optimal with respect to the repair bandwidth but not with respect to the computational complexity. 
Motivated by this, in \cite{5394538}, Rashmi \textit{et al}. proposed a simple construction of Minimum Bandwidth Regenerating codes, where a
survivor node is only required to read and transfer the exact
amount of data to the replacement node without extra computation.  
The construction is extended to new storage code called \textit{Fractional Repetition} (FR) code \cite{rr10,6033980}.
A typical FR code is a two-layer code with an outer MDS code and an inner repetition code, where a linear $[n,k,d]$ code is called a Maximum Distance Separable (MDS) code if $k=n-d+1$. 
Initially, the FR code is defined on homogeneous DSS (DSS with symmetric parameters). 
In such systems, each node has the same amount of encoded packets and the same number of copies of each packets \cite{rr10}. 
In many cases, heterogeneous systems are preferred, in which the node storage capacity and the packet replication factor need not be uniform. 
In \cite{7312417}, an \textit{Adaptive Fractional Repetitions} (Adaptive FR) code is introduced that adapts new parameters of an another FR codes effectively. 
Hence, Adaptive FR codes accept system changes such as disk failure quickly. 
In \cite{7366761}, the minimum distance on an FR code is defined as the minimum number of nodes such that if the nodes fail simultaneously then the system can not recover the complete file. 

For the FR code, the asymmetric node storage capacity helps to reduce the total storage cost and the asymmetric replication factor helps for the availability of favorable packets, where favorable packets are the packets which are downloaded with more frequency.
The asymmetric repair degree and repair bandwidth help to reduce the repair cost of a node in the FR code.
FR codes on asymmetric parameters are studied in \cite{6763122,iet:/content/journals/10.1049/iet-com.2014.1225,7458383,DBLP:journals/corr/abs-1302-3681,DBLP:journals/corr/abs-1303-6801}

In this work, an FR code with asymmetric parameters (node storage capacity, packet replication factor, repair degree, and repair bandwidth) are considered. 
In general, it is shown that any FR code is equivalent to a hypergraph. 
The equivalence between a hypergraph and an FR code leads to the several interesting properties of the FR code.
In general, we provide the existence condition for FR codes using the correspondence.
Bounds on the node storage capacity and the packet replication factor, are obtained for various FR codes. 
A new bound on the replication factor, the number of nodes and the number of packets are obtained such that all the FR codes satisfying the bound are universally good.
Using the correspondence, we have obtained some new bounds.
FR codes meeting the bounds are unknown in the literature. 
It is observed that an FR code constructed using a linear hypergraph is universally good code.
A recursive algorithm is given to construct linear hypergraph.
FR codes which are equivalent to the hypergraph constructed from the algorithm are adaptive FR codes and universally good with asymmetric parameters. 
The universally good adaptive FR codes with asymmetric parameters can adapt system changes effectively. 
The codes allow updating systems without the need for frequent re-configurations.

The paper is organized as follows. 
In Section \ref{sec:2}, we collect some preliminaries on FR codes and hypergraphs with their parameters and properties. 
In Section \ref{sec:3}, the correspondence between hypergraphs and FR codes are studied. 
In the same section, construction of universally good adaptive FR code on heterogeneous DSS is discussed. 
Using the correspondence, various bounds on the parameters of FR codes are obtained in the same section. 
At the end of the Section, the comparison graph is plotted between the reconstruction degree and file size for the FR codes.
Section \ref{sec:4} concludes the paper with some general remarks.

\section{Preliminaries}\label{sec:2}
 In this section, background on FR codes and hypergraphs are given. 
\subsection{Background on Fractional Repetition Codes}

Let a file be divided into $M$ distinct information packets, where information packet symbols are associated with a finite field $\mathbb{F}_q$. 
By using ($\theta$, $M$) MDS code on the information packets, the $M$ information packets are encoded into $\theta$ distinct packets $P_1$, $P_2$, $\ldots$, $P_{\theta}$. 
Because of the property of MDS code, the complete file information can be extracted from any $M$ packets. 
Each packet $P_j$ ($j=1,2,\ldots,\theta$) is replicated $\rho_j$ times in the system. 
All the packets are distributed on $n$ distinct nodes $U_1$, $U_2$, $\ldots$, $U_n$ such that a node $U_i$ ($i=1,2,\ldots,n$) contains $\alpha_i$ packets. 
For the FR code, the maximum node storage capacity and the maximum replication factor are given by $\alpha$ = $\max\{\alpha_i:i=1,2,\ldots,n\}$ and $\rho = \max\{\rho_j:j=1,2,\ldots,\theta\}$ respectively. 
The FR code is denoted by $\mathscr{C}:(n,\theta,\vec{\alpha},\vec{\rho})$, where storage vector is $\vec{\alpha}=(\alpha_1\ \alpha_2\ldots\alpha_n)$ and replication vector is $\vec{\rho}=(\rho_1\ \rho_2\ldots\rho_\theta)$.

\begin{example}
Let a file be splitted into $5\ (=M)$ information packets and encoded into $6\ (=\theta)$ distinct packets $P_1$, $P_2$, $P_3$, $P_4$, $P_5$ and $P_6$ using a $(6,5)$ MDS code. 
The $13$ replicas of the $6\ (=\theta)$ distinct packets are stored in the $5$ distinct nodes $U_1$, $U_2$, $U_3$, $U_4$ and $U_5$ such that $U_1$ = $\{P_1,P_2,P_3,P_4\}$, $U_2$ = $\{P_1,P_5,P_6\}$, $U_3$ = $\{P_2,P_5\}$, $U_4$ = $\{P_3,P_6\}$ and $U_5$ = $\{P_4,P_6\}$.  
The node packet distribution of the FR code 
is illustrated in Table \ref{distribution example}. 
Hence, storage vector is $\vec{\alpha}$ = $(4\ 3\ 2\ 2\ 2)$ 
and replication vector is $\vec{\rho}$ = $(2\ 2\ 2\ 2\ 2\ 3)$. 
In the FR code, $\rho=\max\{\rho_j:j=1,2,\ldots,6\}=3$ and $\alpha=\max\{\alpha_i:i=1,2,\ldots,5\}=4$. Hence, the FR code is denoted by $\mathscr{C}:(n=5,\theta=6,\vec{\alpha} = (4\ 3\ 2\ 2\ 2),\vec{\rho} = (2\ 2\ 2\ 2\ 2\ 3))$.
Note that in this example file size $M$ and the number of nodes $n$ are same and this is not always true.
\label{example}
\end{example}

\begin{table}[ht]
\caption{Distribution of replicated packets on nodes for the FR code $\mathscr{C}:(n=5,\theta=6,\vec{\alpha} = (4\ 3\ 2\ 2\ 2),\vec{\rho} = (2\ 2\ 2\ 2\ 2\ 3))$.}
\centering 
\begin{tabular}{|c||c|}
\hline
$U_1$& $P_1$, $P_2$, $P_3$, $P_4$         \\
\hline \hline
$U_2$& $P_1$, $P_5$, $P_6$                \\
\hline \hline
$U_3$& $P_2$, $P_5$                       \\
\hline \hline 
$U_4$& $P_3$, $P_6$                       \\
\hline \hline 
$U_5$& $P_4$, $P_6$                       \\
\hline  
\end{tabular}
\label{distribution example}
\end{table}

For an FR code, if a data collector connects any $k$ nodes and downloads $M$ distinct packets then the file can be retrieved by using the MDS property.
The parameter $k$ is called the reconstruction degree of the FR code.
Note that the reconstruction degree of the FR code (as given in Example \ref{example}) is $k=3$ for $M=5$.
In an FR code $\mathscr{C}:(n,\theta,\vec{\alpha},\vec{\rho})$, a failed node $U_i$ can be repaired by replacing a new node $U_i'$, where both the nodes $U_i$ and $U_i'$ carry the same information. 
For the new node, the information packets are downloaded from some active nodes called helper nodes. 
The repair degree of a failed node $U_i$ is the cardinality of the set of the helper nodes.
For a node $U_i$ in an FR code, more than one such helper node sets exist with different repair degree.
For the node $U_i$, the maximum repair degree is $d_i$ and $d$ = $\max\{d_i:i=1,2,\ldots,n\}$.
In the FR code $\mathscr{C}:(5,6,4,3)$ (Example \ref{example}), if node $U_2$ fails then the new node $U_2'$ can be created by downloading packets from nodes of set $\{U_1,U_3,U_4\}$ or set $\{U_1,U_3,U_5\}$. 
Hence, $d_2=3$ and $d=\max\{d_i:i=1,2,3,4,5\}=4$. 

Consider a set $S\subset\{P_j:j=1,2,\ldots,\theta\}$ and a set $T\subset \{ U_1,U_2,\dots,U_n\}$. An FR code $\mathscr{C}':(n', \theta', \vec{\alpha'}, \vec{\rho'})$ adapts the FR code $\mathscr{C}:(n, \theta, \vec{\alpha}, \vec{\rho})$, if $n=n'+|T|$, $\theta=\theta'+|S|$ and $U_i'=U_i\backslash T$ for each $i\in S$. 
Hence, the FR code $\mathscr{C}$ is adaptive FR code. 
Note that an adaptive FR code can be updated efficiently with higher parameters \cite{7312417}. 

In an FR code $\mathscr{C}:(n,\theta,\vec{\alpha},\vec{\rho})$, if $M(k)$ is the maximum number of distinct packets that guarantee to deliver to any user connected with any $k$ nodes of the FR code then $M(k):= \min\left\lbrace \left| \cup_{i\in I}U_i\right|: |I|=k, I\subset\right.$ $\left.\{1,2,\ldots,n\} \right\rbrace$.
Note that $M\leq M(k)$.    
For the FR code as given in Example \ref{example}, one can find that $M(3)=5$. 

In \cite{rr10}, an FR code with symmetric parameters (uniform replication factor and identical node storage capacity) is called \textit{universally good} if the maximum file size $M(k)\geq k\alpha-\binom{k}{2}$, where $k\alpha-\binom{k}{2}$ is MBR capacity \cite{5550492}.  
Universally good FR codes with symmetric parameters are studied in \cite{8309370}.
An FR code $\mathscr{C}:(n, \theta, \vec{\alpha}, \vec{\rho})$ with $\rho_j=\rho$ ($j=1,2,\ldots,\theta$), $\alpha_1\leq \alpha_2\leq\ldots\leq \alpha_n$ and $|U_i\cap U_j|\leq1$ ($1\leq i<j\leq n$) satisfies the inequality $M(k)\geq\sum_{i=1}^k \alpha_i-\binom{k}{2}$ \cite{6763122}.
It is easy to observe that any FR code $\mathscr{C}:(n,\theta,\vec{\alpha},\vec{\rho})$ with $|U_i\cap U_j|\leq 1$ satisfies $M(k)\geq \sum_{i=1}^k\alpha_i-\binom{k}{2},$ where $\alpha_i\leq\alpha_j$ ($1\leq i<j\leq n$).
Hence, an FR code $\mathscr{C}:(n, \theta, \vec{\alpha}, \vec{\rho})$ with asymmetric parameters is called \textit{universally good}, if  $|U_i\cap U_j|\leq1$ ($1\leq i<j\leq n$).

In \cite{7366761}, the minimum distance of FR code is defined as the minimum number of nodes such that if these nodes fail then the active nodes cannot recover the complete file.
More precisely, the minimum distance of an FR code $\mathscr{C}:(n,\theta,\vec{\alpha},\vec{\rho})$ is $d_{min}=\min\{|S|:S\subset [n]\ s.t.\ M(k)>\cup_{i\in[n]\backslash S}U_i\}$, where $[n]=\{1,2,\ldots,n\}$.

\subsection{Background on Hypergraphs}

For a finite set $V=\{v_1,v_2,\ldots,v_n\}$, the pair $(V,\mathcal{E})$ is called a hypergraph with hypervertex set $V$ and hyperedge set $\mathcal{E}$, where $\mathcal{E}\subseteq\{E: E\subseteq V\}$ is a family of some subset of $V$ \cite{Berge89a}.
For a hypergraph  $(V,\mathcal{E})$, $|E|$ is called the size of the hyperedge $E\in\mathcal{E}$ and $|\mathcal{E}(v)|$ is the degree of the hypervertex $v\in V$, where $\mathcal{E}(v)$ = $\{E:\ v\in E\}$. 
An example of such hypergraph is shown in Figure \ref{hypergraph example}. 
In this hypergraph $  (V,\mathcal{E})$, hypervertices are shown by simply bold dots and the hyperedges are represented by the covered area which contains some hypervertices. 
In the particular example, there are $7$ hypervertices and $4$ hyperedges such that $\mathcal{E}(v_1)=\{E_1\}$, $\mathcal{E}(v_2)=\phi$, $\mathcal{E}(v_3)=\{E_3\}$, $\mathcal{E}(v_4)=\{E_1,E_2\}$, $\mathcal{E}(v_5)=\{E_1,E_2,E_3\}$, $\mathcal{E}(v_6)=\{E_3\}$ and $\mathcal{E}(v_7)=\{E_4\}$.

\tikzstyle{vertex} = [fill,shape=circle,node distance=60pt]
\tikzstyle{edge} = [fill,opacity=.5,fill opacity=.3,line cap=round, line join=round, line width=30pt]
\pgfdeclarelayer{background}
\pgfsetlayers{background,main}
\begin{figure}[ht]
    \centering
    \begin{tikzpicture}
    \node[vertex,label=above:\(v_1\)] (v1) {};
    \node[vertex,right of=v1,label=above:\(v_4\)] (v2) {};
    \node[vertex,right of=v2,label=above:\(v_5\)] (v3) {};
    \node[vertex,below of=v1,label=above:\(v_7\)] (v4) {};
    \node[vertex,right of=v4,label=above:\(v_6\)] (v5) {};
    \node[vertex,label=above:\(v_3\)] (v6) at (2.8,-1.3) {};
    \node[vertex,label=above:\(v_2\)] (v7) at (1.2,-1.3) {};
    
    \begin{pgfonlayer}{background}
    \draw[edge,color=green,line width=35pt] (v1) -- (v2) -- (v3);
    \begin{scope}[transparency group,opacity=.5]
    \draw[edge,opacity=1,color=purple] (v3) -- (v5) -- (v3);
    \fill[edge,opacity=1,color=purple] (v3.center) -- (v5.center) -- (v3.center);
    \end{scope}
    \draw[edge,color=yellow,line width=25pt] (v2) -- (v3);
    \draw[edge,color=red,line width=40pt] (v4) -- (v4);
    \end{pgfonlayer}
    
    \node[label=right:\(E_1\)]  (e1) at (0.6,0) {};
    \node[label=right:\(E_2\)]  (e1) at (2.5,0) {};
    \node[label=right:\(E_3\)]  (e1) at (3,-1) {};
    \node[label=right:\(E_4\)]  (e1) at (0,-1.7) {};
    
    \end{tikzpicture}
    \caption{A hypergraph $  (V,\mathcal{E})$ is shown, where $V=\{v_i:i\in[7]\}$ and $\mathcal{E}=\{E_1=\{v_1,v_4,v_5\},E_2=\{v_4,v_5\},E_3=\{v_3,v_5,v_6\},E_4=\{v_7\}\}$.}
    \label{hypergraph example}
\end{figure}
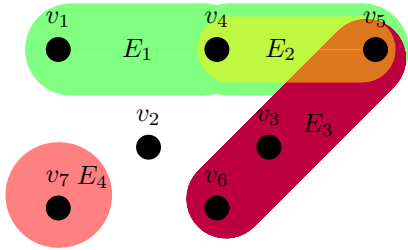

A hypergraph  $(V',\mathcal{E}')$ is called a sub-hypergraph of a hypergraph  $(V,\mathcal{E})$ if $V'\subseteq V$ and $\mathcal{E}'\subseteq\{E'\neq\phi:E'\subseteq E\cap V', E\in\mathcal{E}\}$. 
For an example, a hypergraph  $(V',\mathcal{E}')$ with $V'=\{v_1,v_3,v_4,v_5,v_6\}$ and $\mathcal{E}'=\{E'_1=\{v_1,v_4,v_5\},E'_2=\{v_4\},E'_3=\{v_3,v_5\}\}$, is a sub-hypergraph of the hypergraph  $(V,\mathcal{E})$ (as shown in Figure \ref{hypergraph example}).
For a hypergraph  $(V,\mathcal{E})$, hypergraph  $(V',\mathcal{E}')$ is called sub-hypergraph induced by the set $V'\subset V$ if $\mathcal{E}'=\{E'\neq\phi:E'= E\cap V', E\in\mathcal{E}\}$.
A hypergraph is called connected hypergraph if there exists a sub-hypergraph with the same hypervertex set such that the sub-hypergraph is isomorphic to a connected graph. The hypergraph  $(V,\mathcal{E})$ (Figure \ref{hypergraph example}) is not a connected hypergraph since there is an isolated hypervertex.

Consider a hypergraph  $(V,\mathcal{E})$ without isolated hypervertices. 
For a finite set $V^*=\{v^*_1,v^*_2,\ldots,v^*_{|\mathcal{E}|}\}$, a hypergraph   $(V^*,\mathcal{E}^*)$ is called the dual of the hypergraph  $(V,\mathcal{E})$ if their exist a bijection $g:\mathcal{E}\rightarrow V^*$ such that $E_j=\{g(E_i)=v^*_i:v_j\in E_i\}$, where $j=1,2,\ldots,|V|$. 
For an example, a hypergraph $  (V^*,\mathcal{E}^*)$ with $V^*=\{v_1^*,v_2^*,v_3^*,v_4^*\}$ and $\mathcal{E}^*=\{E_1^*=\{v_1^*\}, E_2^*=\{v_3^*\}, E_3^*=\{v_1^*,v_2^*\}, E_4^*=\{v_1^*,v_2^*,v_3^*\},E_5^*=\{v_4^*\}\}$, is the dual of a hypergraph $  (V,\mathcal{E})$ with $V=\{v_1,v_2,v_3,v_4,v_5\}$ and $\mathcal{E}=\{E_1=\{v_1,v_3,v_4\}$, $E_2 = \{v_3,v_4\},E_3=\{v_2,v_4\},E_4=\{v_5\}\}$, where $g(E_j)=v^*_j$ and $i=1,2,\ldots,5$.

In a hypergraph, if two hyperedges share one node maximum then the hypergraph is called linear hypergraph. 
For a linear hypergraph  $(V,\mathcal{E})$, 
    \begin{equation}
        |\mathcal{E}|\geq \sum_{v\in V}|\mathcal{E}(v)|-\binom{|v|}{2}.
        \label{}
    \end{equation}
A sub-hypergraph of a linear hypergraph is also a linear hypergraph. For a positive integer $k$ (not more than the size of the hypergraph), one can find a bound on the size of sub-hypergraph induced by any $k$-element subset. 
For $V=\{v_i:i=1,2,\ldots,n\}$, consider a linear hypergraph  $(V,\mathcal{E})$ with $|\mathcal{E}(v_i)|\leq|\mathcal{E}(v_j)|$, for $i<j$ and $i,j=1,2,\ldots,n$. For a given positive integer $k<n$, let   $(V'\subset V,\mathcal{E}')$ be a induced sub-hypergraph of the hypergraph  $(V,\mathcal{E})$ with $\mathcal{E}'=\{E'\neq\phi: E'=E\cap V', E\in\mathcal{E}\}$ and $|V'|=k$. Then
    \begin{equation}
        \min\{|\mathcal{E}'|:V'\subset V,|V'|=k\}\geq \sum_{i=1}^k|\mathcal{E}(v_i)|-\binom{k}{2}.
    \end{equation}
\section{Hypergraphs and Fractional Repetition Codes}\label{sec:3}

In this Section, A bijection (called Hyper-FR mapping) between a hypergraph and an FR code is studied. 
Using the Hyper-FR mapping, the relation between the parameters of FR codes and hypergraphs are studied at the end of the Section. 

In an FR code, the distribution of nodes and packets are similar to the incidence of hypervertices and hyperedges in a hypergraph. 
An FR code can be represented by a hypergraph $(V,\mathcal{E})$ with $n$ hypervertices and $\theta$ hyperedges.
The following fact is apparent.
\begin{fact}
An FR code with $n$ nodes, $\theta$ packets, node storage capacity $\alpha_i$ ($i=1,2,\ldots,n$) and packet replication factor $\rho_j$ ($j=1,2,\ldots,\theta$) is equivalent to a hypergraph $(V,\mathcal{E})$ with $|V|=n$ and $|\mathcal{E}|=\theta$ such that $|\mathcal{E}(v_i)|$ = $\alpha_i$ for $v_i\in V$, and $|E_j|$ = $\rho_j$ for $E_j\in\mathcal{E}$.
Hence, $\rho = \max\{|E|:E\in\mathcal{E}\}$ and $\alpha=\max\{|\mathcal{E}(v_i)|:i=1,2,\ldots,n\}$. 
\label{Hypergraph and FR code}
\end{fact}
\begin{example}
As shown in Figure \ref{example figure}, consider a hypergraph  $(V,\mathcal{E})$ with $V=\{v_i:i=1,2,3,4\}$ and $\mathcal{E}$ = $\{E_j:j=1,2,3,4,5,6\}$, where $E_1$ = $\{v_1,v_2\}$, $E_2$ = $\{v_1,v_3\}$, $E_3$ = $\{v_1,v_4\}$, $E_4$ = $\{v_2,v_3\}$, $E_5$ = $\{v_2,v_4\}$ and $E_6$ = $\{v_3,v_4\}$. 
For the hypergraph, $\mathcal{E}(v_1)$ = $\{E_1,E_2,E_3\}$, $\mathcal{E}(v_2)$ = $\{E_1,E_4,E_5\}$, $\mathcal{E}(v_3)$ = $\{E_2,E_4,E_6\}$ and $\mathcal{E}(v_4)$ = $\{E_3,E_5,E_6\}$. 
As given in Figure \ref{example figure}, consider an FR code $\mathscr{C}:(4,6,(3\ 3\ 3\ 3),(2\ 2\ 2\ 2\ 2\ 2))$ with $4$ nodes and $6$ such that $U_1$ = $\{P_1,P_2,P_3\}$, $U_2$ = $\{P_1,P_4,P_5\}$, $U_3$ = $\{P_2,P_4,P_6\}$ and $U_4$ = $\{P_3,P_5,P_6\}$. Note that there exist the bijection $\varphi:\mathcal{E}\to\{P_j:j=1,2,\ldots,6\}$ such that $\varphi(E_j)=P_j$ for each $j=1,2,\ldots,|\mathcal{E}|$. Hence, $\varphi(\mathcal{E}(v_i))$ = $\{P_j=\varphi(E_j):E_j\in\mathcal{E}, v_i\in E_j\}$.
\label{Hyper-FR mapping example}
\end{example}


          \tikzstyle{line} = [draw, -]
          \tikzstyle{cloud} = [draw, circle,fill=red!20, node distance=2cm, minimum height=2em]
            \tikzstyle{block 1} = [rectangle, draw, fill=red!20, text width=6em, text centered, rounded corners, minimum height=1em, node distance=1.7cm]
          \tikzstyle{block 2} = [rectangle, draw, fill=red!20, text width=6em, text centered, rounded corners, minimum height=1em, node distance=0.7cm]
          \tikzstyle{block 3} = [rectangle, draw, fill=red!20, text width=6em, text centered, rounded corners, minimum height=1em, node distance=0.7cm]
           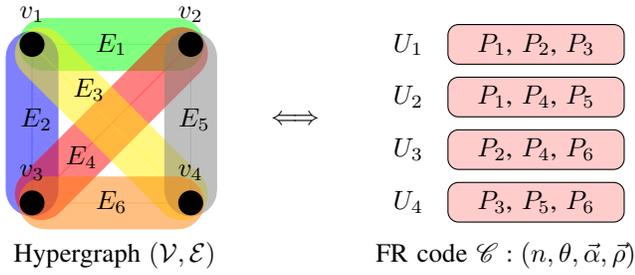
\begin{figure}
           \centering
          \begin{tikzpicture}[node distance = 2cm, auto]
          
         \node[vertex,label=above:\(v_1\)] (v1) {};
         \node[vertex,right of=v1,label=above:\(v_2\)] (v2) {};
         \node[vertex,below of=v1,label=above:\(v_3\)] (v3) {};
        \node[vertex,right of=v4,label=above:\(v_4\)] (v4) {};
            
    \begin{pgfonlayer}{background}
    \draw[edge,color=green,line width=20pt] (v1) -- (v2);
    \draw[edge,color=blue,line width=20pt] (v1) -- (v3);
    \draw[edge,color=red,line width=20pt] (v2) -- (v3);
    \draw[edge,color=yellow,line width=20pt] (v1) -- (v4);
    \draw[edge,color=gray,line width=20pt] (v2) -- (v4);
    \draw[edge,color=orange,line width=20pt] (v3) -- (v4);
    \end{pgfonlayer}
    
    \node[label=right:\(E_1\)]  (e1) at (0.6,0) {};
    \node[label=right:\(E_2\)]  (e1) at (-0.4,-1) {};
    \node[label=right:\(E_3\)]  (e1) at (0.3,-0.6) {};
    \node[label=right:\(E_4\)]  (e1) at (0.2,-1.5) {};
    \node[label=right:\(E_5\)]  (e1) at (1.7,-1) {};
    \node[label=right:\(E_6\)]  (e1) at (0.6,-2.1) {};
    
          \draw (3.5,-1) node {$\Longleftrightarrow$};
         \node (u1) at (5,0) {$U_1$};
          \node [block 1, right of=u1] (u3) {$P_1$, $P_2$, $P_3$};
          \node [block 2, below of=u3] (u4) {$P_1$, $P_4$, $P_5$};
          \node [block 2, below of=u4] (u5) {$P_2$, $P_4$, $P_6$};
          \node [block 2, below of=u5] (u6) {$P_3$, $P_5$, $P_6$};
          \draw (5,-0.7) node {$U_2$};
          \draw (5,-1.4) node {$U_3$};
          \draw (5,-2.1) node {$U_4$};
          \draw (1.1,-2.8) node {Hypergraph $(\mathcal{V},\mathcal{E})$};
          \draw (6.3,-2.8) node {FR code $\mathscr{C}:(n,\theta,\vec{\alpha},\vec{\rho})$};
          \end{tikzpicture}
          \caption{An example of a Hypergraph $(\mathcal{V},\mathcal{E})$ which is corresponding to an FR code $\mathscr{C}:(n=4,\theta=6,\alpha=3,\rho=2)$, 
          where $\mathcal{V}$=$\{v_1,v_2,v_3,v_4\}$ and $\mathcal{E}$ = $\left\lbrace \{v_1,v_2\}, \{v_1,v_3\},\{v_1,v_4\},\{v_2,v_3\},\{v_2,v_4\},\{v_3,v_4\}\right\rbrace $.}
          \label{example figure}
          \end{figure}
For a hypergraph $(V,\mathcal{E})$, let $(V',\mathcal{E}')$ be a sub-hypergraph induced by $V'\subset V$.
The minimum distance of the FR code is $d_{min}$ = $\min\{|V\backslash V'|:V'\subset V\ s.t.\ |\mathcal{E}'|<\theta\}$. 
For the given positive integer $k<n$, the maximum file size stored in the FR code is $M(k)=\min\{|\mathcal{E}'|:V'\subset V\ s.t.\ |V'|=k\}$.
\begin{remark}
For a hypergraph  $(V\neq\phi,\mathcal{E})$, an FR code $\mathscr{C}:(n,\theta, \alpha, \rho)$ exists if and only if $|\mathcal{E}(v)|>0$ ($\ \forall\ v\in V$) and $|E|>0$ ($\ \forall\ E\in\mathcal{E}$).
\end{remark}

The necessary and sufficient condition for the existence of an FR code is given in the following remark.
\begin{remark}
Consider two vectors ($\rho_1\ \rho_2\ldots\rho_\theta$) and $(\alpha_1$ $\alpha_2\ldots\alpha_n)$ on integers such that $\alpha=\alpha_1\geq$ $\alpha_2\geq$ $\ldots$ $\geq\alpha_n$ for $j=1,2,\ldots,\theta$ and $i=1,2,\ldots,n$. 
An FR code $\mathscr{C}:(n,\theta,\vec{\alpha},\vec{\rho})$ exists with $|U_i|$ = $\alpha_i$ for $i=1,2,\ldots, n$ and replication factor of packet $P_j$ is $\rho_j$ for $j=1,2,\ldots,\theta$ if
    \begin{enumerate}
        \item\label{constraint 1} $\sum\limits_{j=1}^{\theta}\rho_j=\sum\limits_{i=1}^n\alpha_i$ and
        \item\label{constraint 2} $\sum\limits_{j=1}^{\theta}\min\{\rho_j,m\}\geq\sum\limits_{i=1}^m\alpha_i\ (\mbox{ for } m<n, m\in\mathbb{Z})$.
    \end{enumerate}
\end{remark}

Note that if the parameters of an FR code do not satisfy one of the two constraints, then the FR code does not exist.

\begin{fact}
An FR code is universally good if and only if the FR code is a linear hypergraph.
\end{fact}
\begin{construction}
For a non-empty hypervertex set $V$ and set $\mathcal{E}\subset\{E:E\subset V\}$, a linear hypergraph $(V,\mathcal{E})$ is a universally good FR code with $|V|$ nodes, $|\mathcal{E}|$ packets, node storage capacity $|\mathcal{E}(v_i)|$ ($v_i\in V$ for $i=1,2,\ldots,|V|$) and packet replication factor $|E_j|$ ($E_j\in\mathcal{E}$ for $j=1,2,\ldots,\mathcal{E}$).
\end{construction}
Recall that any sub-hypergraph (or induced sub-hypergraph) of a linear hypergraph is also a linear hypergraph. 
It motives us for recursive construction of a linear hypergraph by adding a hypervertex and/or adding a hyperedge into a linear hypergraph as per the following three facts. 
\begin{fact}
Consider a linear hypergraph ($V,\mathcal{E}$).
For a set $E\subset V$, the hypergraph ($V,\mathcal{E}\cup\{E\}$) is a linear hypergraph if and only if $|E\cap E'|\leq1$ for each $E'\in\mathcal{E}$.
\label{fact hyperedge}
\end{fact}
\begin{fact}
Consider a linear hypergraph ($V,\mathcal{E}$).
For a hypervertex $v\notin V$ and a set $E\subset V$, the hypergraph ($V\cup\{v\},\mathcal{E}\cup\{E\cup\{v\}\}$) is a linear hypergraph if and only if $|E\cap E'|\leq1$ for each $E'\in\mathcal{E}$.
\label{fact hypervertax}
\end{fact}
\begin{fact}
Consider a linear hypergraph ($V,\mathcal{E}$).
For a set $E\subset V$ and a hyperedge $E'\in\mathcal{E}$, let $E'\subset E$.
The hypergraph ($V,(\mathcal{E}\backslash\{E'\})\cup\{E\}$) is a linear hypergraph if and only if $|E\cap E''|\leq1$ for each $E''\in\mathcal{E}$, where $E''\neq E'$.
\label{fact induced sub-hypergraph}
\end{fact}

For example, one can construct the FR code (Figure \ref{example figure}) recursively as given in Table \ref{recursive FR code}. 
In the example, a hypervertex $v_4$ and a hyperedge $E_3$ are added into the initial hypergraph for the step $2$ using the Fact \ref{fact hypervertax}.
In the step $3$ and $4$, hyperedges $E_5$ and $E_6$ are added into the hypergraph.
\begin{remark}
	In Fact \ref{fact hyperedge}, a linear hypergraph is constructed by adding a hyperedge into a linear hypergraph. 
	In corresponding FR codes, a universally good FR code is constructed by adding a new packet into a universally good FR code.  
\end{remark}
\begin{remark}
In Fact \ref{fact hypervertax}, a linear hypergraph is constructed by adding a hypervertex into a linear hypergraph. 
In corresponding FR codes, a universally good FR code is constructed by adding a new node into a universally good FR code.  
\end{remark}
\begin{remark}
	In Fact \ref{fact induced sub-hypergraph}, a linear hypergraph is constructed by replacing a hyperedge with a new hyperedge such that the new hyperedge contains all the hypervertices of the hyperedge. 
	In corresponding FR codes, a universally good FR code is constructed by increasing the replication factor of a packet in a universally good FR code.  
\end{remark}

For an FR code, the dual FR code is equivalent to the dual hypergraph. 
For FR codes with symmetric parameters, the file size of dual FR code is studied in \cite{8277971}.
Therefore, the following remark is evident.
\begin{remark}
    The dual of a linear hypergraph is again a linear hypergraph. 
    So, the dual of a universally good FR code is again a universally good.
\end{remark}
\subsection{Bounds on Fractional Repetition Codes}
In the subsection, bounds on parameters of various FR codes are calculated in Lemma \ref{lemma 30}, \ref{lemma 32} and \ref{lemma 18} using properties of hypergraphs.
\begin{lemma}
    For each $j=1,2,\ldots,\theta$, an FR code $\mathscr{C}:(n,\theta,\vec{\alpha},\vec{\rho})$ with $\rho_j=\rho$ ($j=1,2,\ldots,\theta$) connected conditional hypergraph, satisfies the following conditions.
    \begin{enumerate}
        \item $\sum\limits_{i=1}^n\alpha_i \mbox{ is multiple of }\rho$,
        \item $\sum\limits_{i=1}^n\alpha_i\geq\frac{\rho(n-1)}{\rho -1}$,
        \item $\alpha\leq\frac{1}{\rho}\sum\limits_{i=1}^n\alpha_i$.
    \end{enumerate}
    \label{lemma 30}
\end{lemma}
\begin{proof}
Proof follows the Fact \ref{Hypergraph and FR code} and Proposition 2 \cite[Chapter 1]{Berge89a}.
\end{proof}
\begin{lemma}
Let $\mathscr{C}:(n,\theta=2m,\vec{\alpha},\vec{\rho})$ be an FR code with message packets $P_1$, $P_2$,$\ldots$, $P_m$, $P_1'$, $P_2'$,$\ldots$, $P_m$ such that $F_i\cap F_j'$=$\phi$ if and only if $i=j$. Then $\sum\limits_{j=1}^m\binom{|F_j|+|F_j'|}{|F_j|}^{-1}\leq 1$, where $F_j=\{U_i:P_j\in U_i\}$ and $F_j'=\{U_i:P_j'\in U_i\}$.
    \label{lemma 32}
\end{lemma}
\begin{proof}
Proof follows the Fact \ref{Hypergraph and FR code} and Theorem 6 in \cite[Chapter 1]{Berge89a}.
\end{proof}
\begin{lemma}
    An FR code $\mathscr{C}:(n,\theta,\vec{\alpha},\vec{\rho})$ with $U_i\nsubseteq U_j$ for each $i\neq j$ and $i,j=1,2,\ldots,n$, satisfies $\sum\limits_{i=1}^n\binom{\theta}{\alpha_j}^{-1} \leq 1$ and $n\leq\binom{\theta}{\lfloor \theta/2 \rfloor}$.
    \label{lemma 18}
\end{lemma}
\begin{proof}
 Proof follows the Fact \ref{Hypergraph and FR code} and dual statement of Theorem $2$ of \cite[Chapter 1]{Berge89a}.
\end{proof}

Lemma \ref{lemma 32} and \ref{lemma 18} are the new bounds on FR codes with asymmetric parameters. 
\subsection{Bounds on Universally Good Fractional Repetition Codes}
In the subsection, various bounds on parameters of universally good FR code is calculated using bounds on linear hypergraphs.
Using bounds and dual properties of a linear hypergraph, one can observe Lemma \ref{linear FR code 22}, \ref{lemma 1}, \ref{lemma 37} and \ref{dual of lemma 37}.
\begin{lemma}
    An FR code $\mathscr{C}:(n,\theta,\vec{\alpha},\vec{\rho})$ with $|U_i\cap U_j|\leq 1$ (for each $i\neq j$ and $i,j=1,2,\ldots,n$), satisfies $\sum\limits_{j=1}^{\theta}\binom{\rho_j}{2} \leq \binom{n}{2}$.
    \label{linear FR code 22}
\end{lemma}
\begin{proof}
The proof follows the Fact \ref{Hypergraph and FR code} and Theorem $3$ in \cite[Chapter 1]{Berge89a}.
\end{proof}
\begin{lemma}
An FR code $\mathscr{C}:(n,\theta,\vec{\alpha},\vec{\rho})$ with $|U_i\cap U_j|\leq 1$ (for each $i\neq j$ and $i,j=1,2,\ldots,n$), satisfies $\sum\limits_{i=1}^{n}\binom{\alpha_i}{2} \leq \binom{\theta}{2}$.
\label{lemma 1}
\end{lemma}
\begin{proof}
The proof follows the dual statement of the Lemma \ref{linear FR code 22}.
\end{proof}
\begin{lemma}
    An FR code $\mathscr{C}:(n,\theta,\vec{\alpha},\vec{\rho})$ with $|U_i\cap U_j|\leq 1$ (for each $i,j=1,2,\ldots,n$) and $\alpha_m=\alpha$ (for each $m=1,2,\ldots,n$) satisfies $n \leq \frac{\theta(\theta-1)}{\alpha(\alpha-1)}$.
    \label{lemma 37}
\end{lemma}
\begin{proof}
The proof follows the Lemma \ref{linear FR code 22} and $\alpha_m=\alpha$ (for each $m=1,2,\ldots,n$).
\end{proof}
\begin{lemma}
    An FR code $\mathscr{C}:(n,\theta,\vec{\alpha},\vec{\rho})$ with $|U_i\cap U_j|\leq 1$ (for each $i,j=1,2,\ldots,n$) and $\rho_m=\rho$ (for each $m=1,2,\ldots,n$) satisfies $\theta \leq \frac{n(n-1)}{\rho(\rho-1)}$.
    \label{dual of lemma 37}
\end{lemma}
\begin{proof}
The proof follows the Lemma \ref{lemma 1} and $\rho_\rho=\alpha$ (for each $m=1,2,\ldots,\theta$).
\end{proof}

  
\begin{remark}
Each FR code studied in the literature is associated with a specific hypergraph. 
The various FR codes and equivalent hypergraphs are listed in the Table \rom{2} (see the Appendix of the extended version \cite{DBLP:journals/corr/abs-1711-07631}). 
\end{remark}

As per the best knowledge of authors, all the bounds in Lemma \ref{linear FR code 22}, \ref{lemma 1}, \ref{lemma 37} and \ref{dual of lemma 37} are new bounds for the universally good FR codes with asymmetric parameters. 
In \cite{rr10,7422071}, FR codes with symmetric parameters satisfy the bounds in Lemma \ref{linear FR code 22}, \ref{lemma 1} and \ref{lemma 18} with equality.
It would be an interesting task to find more FR codes with asymmetric parameters which satisfy Lemma \ref{lemma 30}, \ref{lemma 32}, \ref{lemma 18}, \ref{linear FR code 22}, \ref{lemma 1}, \ref{lemma 37} and \ref{dual of lemma 37} with equality.


\section{Conclusion}\label{sec:4}
In this work, we have sketched the surface of the very interesting topic by establishing a correspondence between FR code and hypergraph. 
It is shown that each FR code studied in the literature is associated with some conditional hypergraph, where, for the hypergraph, size of hyperedges and degree of hypervertices are all non-zero positive integers. 
A construction of universally good adaptive FR code with asymmetric parameter is discussed in this work.
Bounds on universally good FR codes are obtained from the properties of hypergraphs. 
Some new bounds are found on universally good FR code defined on heterogeneous DSS.  
It would be an interesting work in the future to find the FR codes which satisfy those bounds with equality. 
Analysis of FR codes on generalized hypergraph could be another interesting problem.

\bibliography{cloud}
\bibliographystyle{IEEEtran}

 \section{Appendix}
 The Appendix includes the literature on FR codes, an algorithm to construct adaptive universally good FR code and a list of hypergraphs corresponding to FR codes from literature.
\subsection{Literature on Fractional Repetition Codes}
In the subsection, literature on FR codes are discussed in following three subsections. 
\subsubsection{Classification of constructions of FR codes on homogeneous DSSs}
FR codes are constructed in many papers.
The classification of constructions of FR codes defined on homogeneous DSSs, are following.

\begin{itemize}
	\item \textit{Algebraic Constructions}: 
	Using partial order sets, FR codes are constructed in \cite{8277958}. 
	In \cite{e18120441,e19100563}, FR codes are constructed using difference sets. 
	In \cite{e19100563}, FR codes are constructed using relative difference sets with $\lambda=1$, cyclic relative difference sets and non-cyclic relative difference sets. 
	FR codes are constructed from perfect difference families and quasi-perfect difference families in \cite{e18120441}.
	In the paper \cite{8255008}, two constructions of FR codes are  proposed which are based on circulant permutation matrices and affine permutation matrices. 
	In \cite{6703309}, FR code, called \textit{Cyclic repetition erasure code}, is constructed using circulant permutation matrices. 
	
	\item \textit{Algorithmic Constructions}: 
	In the paper \cite{DBLP:journals/corr/abs-1303-6801}, an algorithm is presented for constructing the node-packet distribution matrix of an FR code, where the node-packet distribution matrix is the matrix representation of the FR code. 
	Using those matrices, FR codes are enumerated up to a given number of system nodes.
	
	\item \textit{Constructions using Combinatorics}: 
	In \cite{rr10}, FR codes are constructed from fano plane, steiner systems, and their dual designs. 
	In \cite{6483351,6810361} and \cite{7422071}, FR codes are constructed using Affine resolvable designs, mutually orthogonal Latin squares and Kronecker product technique. 
	In \cite{7510829}, the construction of the optimal FR codes are proposed by t-designs. 
	In \cite{6912604}, the FR codes are constructed using uniform group divisible designs.
	The FR codes are constructed using regular packing designs in \cite{8309370}. 
	
	\item \textit{Constructions using Finite Geometry}: 
	Constructions of FR codes with the fewest number of storage nodes given the other parameters, based on finite geometries and corresponding bipartite cage graphs are considered in \cite{6120326}. 
	
	\item \textit{Constructions using Probabilistic Approach}:
	Randomized FR codes are constructed in \cite{6033980} using the balls-and-bins model.
	
	\item \textit{Graph Theory Constructions}: 
	In \cite{rr10,6120326,7118709,6570830,7510829}, FR codes are constructed using graphs, where the distribution of packets on nodes in an FR code are associated with adjacency of edges with vertices in respective graph. 
	In \cite{rr10}, FR codes are constructed from complete graphs and regular graphs.
	For given parameters and the fewest number of storage nodes, FR codes are constructed using bipartite cage graphs in \cite{6120326}. 
	In \cite{7118709}, FR codes are constructed using Tur$\grave{a}$n graph and cage graphs.
	Construction of FR codes based on regular graphs with a given girth, is considered in \cite{6570830}.
	In \cite{7510829}, Xu \textit{et al.} have looked FR code as biregualr graph, and shown that finding FR code with $k$ fault-tolerance and the minimum redundancy, is the same as the \textit{Zarankiewicz} problem in graph theory. 
	
	\item \textit{Other Constructions}: In \cite{DBLP:journals/corr/PrajapatiG16}, FR codes are constructed using ring construction. 
	In the construction nodes are placed in a circle and all the replicated packets are distributed to the nodes in specific order. 
\end{itemize}

For the given parameters, the existence of FR codes are investigated in \cite{DBLP:journals/corr/abs-1201-3547}.
Algorithms are presented in \cite{DBLP:journals/corr/abs-1305-4580} which compute the reconstruction degree $k$ and the repair degree $d$ of the FR code. 
For an FR code, the dual bound on file size is investigated in \cite{8277971}.

\subsubsection{FR code constructions with additional system properties}
Some FR codes with additional system property, are also studied in literature. 
Those FR codes are following.
\begin{itemize}
	\item In \cite{6570830}, the minimum distance of FR code is defined as the minimum number of nodes such that if these nodes fail then the active nodes can not recover the complete file.
	FR codes are constructed in \cite{7366761} using t-design. 
	Construction of FR codes based on regular graphs with a given girth, in particular cages, and analysis of their minimum distance is considered in \cite{6570830}.
	The FR code follows a Singleton-like bound on the minimum distance \cite{7422071}.
    For an FR code $\mathscr{C}:(n,\theta,\vec{\alpha},\vec{\rho})$ on an $(n,k,d)$-DSS, the minimum distance 
	\begin{equation}
	d_{min}\leq n-\left\lceil\frac{M(k)}{\alpha}\right\rceil+1.
	\label{chap 2 min dis bound 1}
	\end{equation}

	\item In \cite{7312417}, \textit{Adaptive} FR code is constructed, where the adaptive FR code can adapt the system changes such as disk failure. 
	The Adaptive FR code is obtained from symmetric design by removing some points from the blocks.
	\item In \cite{DBLP:conf/icmcta/Silberstein14,7005805,7282815} and \cite{7118709} Fractional Repetition Batch (FRB) code is studied which provides uncoded repairs and parallel reads of subsets of stored packets. 
	An FR code is a FRB code if any batch of $t$ information packets can be decoded by reading at most one symbol from each node.
	The FRB codes are constructed using regular graph and complete bipartite graph. 
	\item In \cite{7458387,LFRkics,MiYoungNAM2017}, the Locally Repairable FR codes are constructed using graphs, where an FR code is called \textit{Locally Repairable} FR code if $k>d$. 
	The Locally Repairable FR code constructions are proposed based on the symmetric design in \cite{7558231}.
	Locally recoverable FR code are constructed using Kronecker product technique in \cite{6810361}.
	For a Locally Repairable FR code, note that each failed storage node can have multiple local repair alternatives in the system. 
	In \cite{6570830}, Locally recoverable FR code are constructed using regular graph. 
	In \cite{6620277}, Locally Repairable FR codes are constructed using $t$-design. 
	Locally Repairable FR codes defined on ($n,k,d,\alpha$) DSS with maximum file size $M(k)$ and minimum distance $d_{min}$ satisfy the bound
	\begin{equation}
	d_{min}\leq n-\left\lceil\frac{M(k)}{\alpha}\right\rceil-\left\lceil\frac{M(k)}{d\alpha}\right\rceil+2,
	\end{equation}
	\item In \cite{8272004}, \textit{Pliable} FR codes are studied, where per node storage and per packet replication factor can be adjusted simultaneously. 
	The Pliable FR codes are constructed using Euclidean geometry, circulant permutation matrices, affine permutation matrices, extended Reed Solomon codes, Euler squares, geometry decomposition and zigzag codes in \cite{8272004}. 
    \item In \cite{DBLP:journals/corr/abs-1802-00872}, \textit{Load-Balanced} FR codes are are studied, where multiple node failures in the FR code can be repaired in a sequence by downloading at most one block from any other active node. 
    Note that the property reduces the number of disk reads which need to repair multiple nodes in the system. 
\end{itemize}
\subsubsection{Construction of FR codes with general parameters}
FR codes are generalized into following directions by dropping uniform constraints on parameters.
\begin{itemize}
	\item FR codes are generalized to \textit{Weak} \cite{DBLP:journals/corr/abs-1302-3681,DBLP:journals/corr/PrajapatiG16}, \textit{General} \cite{8094864,6763122} or \textit{Irregular} \cite{6804948} FR codes, where different number of packets are stored in each node.
	Those FR codes are constructed using partial regular graphs \cite{DBLP:journals/corr/abs-1302-3681} and group divisible designs \cite{6763122}. 
    In \cite{DBLP:journals/corr/PrajapatiG16}, FR codes are constructed using ring construction, where all the replicated packets are distributed to nodes in specific order.
	In \cite{8094864}, the FR code is constructed using packings, coverings, and pairwise balanced designs.
	In \cite{6804948}, FR codes are studied as uniform hypergraphs.   
	For the FR code on ($n,k,d$) DSS, the maximum file size 
	\begin{equation}
	M(k)\geq\sum_{i=1}^k d_i-\binom{k}{2},
	\end{equation}
	where $d_1\leq d_2\leq\ldots\leq d_n$ \cite{6763122}. 
	
	\item FR codes are generalized to \textit{variable} \cite{6811237} or \textit{heterogeneous} \cite{iet:/content/journals/10.1049/iet-com.2014.1225} FR codes, where each packet has different replication factor.
	The specific FR codes are constructed using group divisible designs.
	\item In \cite{7066224}, the \textit{Flexible} FR codes are generalized to FR codes, where node storage capacity of different nodes are different, the replication factor of different packets are different and any two distinct nodes contain at most one common packet.
	For the FR code, the maximum file size $M(k)$ is bounded as
	\begin{equation}
	    \sum_{i=n-k+1}^n\alpha_i-\binom{k}{2}\leq M(k)\leq\sum_{i=1}^k\alpha_i,
	\end{equation}
	where $\alpha_1\leq\alpha_2\leq\ldots\leq\alpha_n$.
	The FR code is constructed using pairwise balance designs and group divisible designs.
	A heuristic construction for the FR code is also given in \cite{7066224}.
\end{itemize}

\subsection{Construction of Adaptive Universally Good Fractional Repetition Codes}
Any sub-hypergraph (or induced sub-hypergraph) of a linear hypergraph is also a linear hypergraph. 
Hence, using recurrence arguments Fact \ref{fact hyperedge}, \ref{fact hypervertax} and \ref{fact induced sub-hypergraph}, one can construct a linear hypergraph for any given number of hypervertices.
A recursive construction (based on Fact \ref{fact hyperedge} and \ref{fact hypervertax}) of linear hypergraph is given in the Algorithm \ref{linear hypergraph construction}.
\begin{algorithm} 
    \caption{Construction of linear hypergraph.} 
    \begin{algorithmic} 
        \REQUIRE A positive integer $n$. 
        \ENSURE A linear hypergraph with $n$ hypervertices.
        \STATE  Initialize: $i=2$, $V=\{v_1,v_2\}$ and $\mathcal{E}=\{E_1\}$, where $E_1=\{v_1,v_2\}$.
        \WHILE{$|V|\leq n$}
        \STATE For some $I_i\subset\{1,\ldots,i\}$, calculate $E_{i+1}=\{v_t:t\in I_i\cup\{v_{i+1}\}\}$ 
        \IF{$|E_{i+1}\cap E_j|\leq1$ for each $j=1,\ldots,i$}
        \STATE Update $V$ by $V\cup\{v_{i+1}\}$
        \STATE Update $i$ by $i+1$
        \ENDIF
        \ENDWHILE
    \end{algorithmic}
    \label{linear hypergraph construction}
\end{algorithm}
In the Algorithm \ref{linear hypergraph construction}, the minimum size among all the hyperedges is $2$ so, the minimum replication factor of the FR code is $2$.
One can control the minimum replication factor for any $\rho_{min}\geq 2$ by initializing $E_1=\{v_1,\ldots,v_{\rho_{min}}\}$.
In the Algorithm \ref{linear hypergraph construction} for each while loop, it is assumed that set $I_i$ exists such that $|E_{i+1}\cap E_j|\leq1$ for each $j=1,\ldots,i$.
Note that, for the hypergraphs constructed from the Algorithm \ref{linear hypergraph construction}, FR codes are universally good adaptive FR codes with asymmetric parameters.
An example of construction of adaptive universally good FR code is given in Table \ref{recursive FR code}, where each step is a while loop.

\subsection{Hypergraphs corresponding to Fractional Repetition codes}
In this subsection, we have listed hypergraphs which are equivalent to existing FR codes in various literature.
Each FR code studied in the literature is associated with a specific hypergraph. 
The various FR codes and equivalent hypergraphs are listed in the Table \ref{FR code hypergraph}. 

\begin{table}[ht]
    \caption{FR codes and respective hypergraphs}
    \centering 
    \begin{tabular}{|l|l|}
        \hline
        FR code                                                           & Corresponding hypergraph                                                         \\        
        $\mathscr{C}:(n,\theta,\vec{\alpha},\vec{\rho})$                  & $(V,\mathcal{E})$                                                         \\        \hline \hline
        
        FR code \cite{rr10,7005805,8277958}                                       & $\rho$-uniform and $\alpha$-regular hypergraph                                            \\        \hline 
        
        HFR code \cite{iet:/content/journals/10.1049/iet-com.2014.1225}   & Linear and $\alpha$-regular hypergraph                                                    \\        \hline 
           
        IFR code \cite{6804948}                                           &  $\rho$-uniform hypergraph                                                                \\        \hline 
        
        VFR code \cite{6811237}                                           &  $\alpha$-regular hypergraph with                                                         \\         
                                                                          &  $|E|\in\{\rho_1,\rho_2\}$, for each $E\in\mathcal{E}$                                    \\        \hline 
        
        GFR code \cite{8094864,6763122}                                   & $\rho$-uniform and linear hypergraph                                                      \\        \hline
        
        WFR code \cite{DBLP:journals/corr/abs-1302-3681}                  & $\rho$-uniform hypergraph                                                                 \\        \hline 
        
        FR code \cite{e18120441,e19100563,8255008}                        & $\rho$-uniform, $\alpha$-regular and                                     \\       
        
        \cite{8309370,8277971,7510829,6120326,7422071}                    & linear hypergraph                                                                                          \\        \hline 
        
         Adaptive FR code \cite{7312417}                                   & $\rho$-uniform $\alpha$-regular hypergraph $s.t.$                                   \\    
                                                                           & induced sub-hypergraph is again                                \\
                                                                           & $\rho'$-uniform $\alpha'$-regular hypergraph                       \\        \hline               
        
        FFR code \cite{7066224}                                           & Linear hypergraph                                                                         \\        \hline 
        
        FR code \cite{DBLP:journals/corr/PrajapatiG16}                    & $\rho$-uniform hypergraph                                                                 \\        \hline 
        
        FR code \cite{7458383}                                            & Hypergraph                                                                                \\        \hline 
        
        FR code \cite{7366761,DBLP:journals/corr/abs-1303-6801}           & $\rho$-uniform and $\alpha$-regular hypergraph                                            \\        \hline 
        
        FR code \cite{7282815},                            & Intersecting, $\rho$-uniform and                                                         \\        
        
        \cite{6483351,6810361}                                           &  $\alpha$-regular hypergraph                                                               \\        \hline
        
        
        Locally repairable                                                & $\rho$-uniform, $\alpha$-regular and                                                     \\        
        
        FR code \cite{7458387,7558231,MiYoungNAM2017,LFRkics}             & and non-linear hypergraph                                                                \\        \hline 
        
        FRB codes \cite{DBLP:conf/icmcta/Silberstein14,7118709,7005805,7282815}   &  $\rho$-uniform and $\alpha$-regular hypergraph                                           \\       \hline 

        
    \end{tabular}
    \label{FR code hypergraph}
\end{table}

\begin{table}[ht]
	\caption{Recursive construction of FR code (Figure \ref{example figure})}
	\centering 
	\begin{tabular}{|l|l|l|l|} 
	\hline
	Step & Hypergraph 	& Adaptive Universally            & Fact       \\ 
	 & ($V,\mathcal{E}$)	& Good FR code            &        \\ 
	\hline
	$1$  & $V=\{v_1,v_2,v_3\}$, 	        &                    &            \\ 
	&$\mathcal{E}=\{E_1,E_2,E_4\}$, 	&                    & Initial           \\ 
	&$E_1=\{v_1,v_2\}$, 	            & $U_1=\{P_1,P_2\}$, & Hyper-    \\ 
	&$E_2=\{v_1,v_3\}$, 	            & $U_2=\{P_1,P_4\}$, & graph \\ 
	&$E_4=\{v_2,v_3\}$ 	            & $U_3=\{P_2,P_3\}$  &            \\ 
	\hline
	$2$&$V=\{v_1,v_2,v_3,v_4\}$	        & &         \\ 
	&$\mathcal{E}=\{E_1,E_2,E_4,E_3\}$, 	&                    &            \\ 
	&$E_1=\{v_1,v_2\}$, 	            & $U_1=\{P_1,P_2,P_3\}$, &     \\ 
	&$E_2=\{v_1,v_3\}$, 	            & $U_2=\{P_1,P_4\}$, & Fact \ref{fact hypervertax} \\ 
	&$E_4=\{v_2,v_3\}$, 	            & $U_3=\{P_2,P_3\}$,  &            \\ 
	&$E_3=\{v_1,v_4\}$ 	            & $U_4=\{P_3\}$  &            \\ 
	\hline
	$3$&$V=\{v_1,v_2,v_3,v_4\}$	        & &         \\ 
	&$\mathcal{E}=\{E_i:i=1,2,3,4,5\}$, 	&                    &            \\ 
	&$E_1=\{v_1,v_2\}$, 	            & $U_1=\{P_1,P_2,P_3\}$, &     \\ 
	&$E_2=\{v_1,v_3\}$, 	            & $U_2=\{P_1,P_4,P_5\}$, & Fact \ref{fact hyperedge} \\ 
	&$E_3=\{v_1,v_4\}$, 	            & $U_3=\{P_2,P_3\}$,  &            \\ 
	&$E_4=\{v_2,v_3\}$, 	            & $U_4=\{P_3,P_5\}$  &            \\
	&$E_5=\{v_2,v_4\}$ 	            &                &            \\
	\hline
	$4$&$V=\{v_1,v_2,v_3,v_4\}$	        & &         \\ 
	&$\mathcal{E}=\{E_i:i=1,2,\ldots,6\}$, 	&                    &            \\ 
	&$E_1=\{v_1,v_2\}$, 	            & $U_1=\{P_1,P_2,P_3\}$, &     \\ 
	&$E_2=\{v_1,v_3\}$, 	            & $U_2=\{P_1,P_4,P_5\}$, & Fact \ref{fact hyperedge} \\ 
	&$E_3=\{v_1,v_4\}$, 	            & $U_3=\{P_2,P_3,P_6\}$,  &            \\ 
	&$E_4=\{v_2,v_3\}$, 	            & $U_4=\{P_3,P_5,P_6\}$  &            \\
	&$E_5=\{v_2,v_4\}$, 	            & (Universally good                &            \\
	&$E_6=\{v_3,v_4\}$ 	            &    FR code (Figure \ref{example figure}))            &            \\
	\hline
	\end{tabular}
	\label{recursive FR code}
\end{table}
\end{document}